\documentclass{llncs}

\usepackage[ruled, vlined, nofillcomment, linesnumbered]{algorithm2e}

\usepackage{cite}
\usepackage{url}
\usepackage{amsmath}
\usepackage{amssymb}
\usepackage{graphicx}
\usepackage[tight]{subfigure}
\usepackage[english]{babel}
\usepackage{booktabs}
\usepackage{nicefrac}
\usepackage{verbatim}

\usepackage{tikz}
\usetikzlibrary{snakes,arrows,shapes,positioning,fit,calc,matrix,trees}
\usepackage{pgfplots}

\usepackage[pdftitle={Discovering Descriptive Tile Trees by Mining Optimal Geometric Subtiles},
           pdfauthor={Nikolaj Tatti and Jilles Vreeken},
		   hidelinks,
		   hyperfootnotes = false]{hyperref}

\newcommand{\FindTile}{\textsc{FindTile}\xspace}
\newcommand{\Mondrian}{\textsc{Stijl}\xspace}

\newcommand{\DB}{D}
\newcommand{\Model}{X}
\newcommand{\Models}{\mathcal{\Model}}

\newcommand{\by}[2]{{\ensuremath{#1\text{\nobreakdash-by\nobreakdash-}#2}}}

\newcommand{\set}[1]{\left\{#1\right\}}
\newcommand{\pr}[1]{\left(#1\right)}
\newcommand{\fpr}[1]{\mathopen{}\left(#1\right)}

\newcommand{\abs}[1]{{\left|#1\right|}}

\newcommand{\enpr}[2]{\pr{#1 ,\ldots , #2}}

\newcommand{\define}{\leftarrow}

\newcommand{\freq}[1]{\mathit{fr}\fpr{#1}}

\newcommand{\numrows}{N}
\newcommand{\numcols}{M}

\newcommand{\tree}[1]{\mathcal{#1}}
\newcommand{\children}[1]{\mathit{children}\fpr{#1}}
\newcommand{\tid}[1]{\mathit{tid}\fpr{#1}}
\newcommand{\ones}[1]{\mathit{p}\fpr{#1}}
\newcommand{\zeroes}[1]{\mathit{n}\fpr{#1}}
\newcommand{\cells}[1]{\mathit{area}\fpr{#1}}
\newcommand{\lc}[1]{\mathit{L}\fpr{#1}}
\newcommand{\hborders}[1]{\mathit{bh}\fpr{#1}}
\newcommand{\tborders}[1]{\mathit{bt}\fpr{#1}}
\newcommand{\hcands}[1]{\mathit{ch}\fpr{#1}}
\newcommand{\tcands}[1]{\mathit{ct}\fpr{#1}}
\newcommand{\hfreq}[1]{\mathit{hfr}\fpr{#1}}
\newcommand{\tfreq}[1]{\mathit{tfr}\fpr{#1}}
\newcommand{\cnt}[1]{\mathit{cnt}\fpr{#1}}
\newcommand{\cost}[1]{\mathit{cost}\fpr{#1}}

\newcommand{\ent}[1]{L\fpr{#1}}

\SetKwComment{tcpas}{\{}{\}}
\SetCommentSty{textnormal}
\SetArgSty{textnormal}
\SetKw{False}{false}
\SetKw{True}{true}
\SetKw{Null}{null}
\SetKwInOut{Output}{output}
\SetKwInOut{Input}{input}
\SetKw{AND}{and}
\SetKw{OR}{or}
\SetKw{Break}{break}

\pgfdeclarelayer{background}
\pgfdeclarelayer{foreground}
\pgfsetlayers{background,main,foreground}

\definecolor{yafaxiscolor}{rgb}{0.3, 0.3, 0.3}

\definecolor{yafcolor1}{rgb}{0.4, 0.165, 0.553}
\definecolor{yafcolor2}{rgb}{0.949, 0.482, 0.216}
\definecolor{yafcolor3}{rgb}{0.47, 0.549, 0.306}
\definecolor{yafcolor4}{rgb}{0.925, 0.165, 0.224}
\definecolor{yafcolor5}{rgb}{0.141, 0.345, 0.643}
\definecolor{yafcolor6}{rgb}{0.965, 0.933, 0.267}
\definecolor{yafcolor7}{rgb}{0.627, 0.118, 0.165}
\definecolor{yafcolor8}{rgb}{0.878, 0.475, 0.686}

\newlength{\yafaxispad}
\newlength{\yaftlpad}
\newlength{\yaflabelpad}
\newlength{\yafaxiswidth}
\newlength{\yafticklen}
\makeatletter
\def\pgfplots@drawtickgridlines@INSTALLCLIP@onorientedsurf#1{}
\makeatother

\newcommand{\yafdrawxaxis}[2]{
	\pgfplotstransformcoordinatex{#1}\let\xmincoord=\pgfmathresult 
	\pgfplotstransformcoordinatex{#2}\let\xmaxcoord=\pgfmathresult 
	\pgfsetlinewidth{\yafaxiswidth} 
	\pgfsetcolor{yafaxiscolor}
	\pgfpathmoveto{\pgfpointadd{\pgfpointadd{\pgfplotspointrelaxisxy{0}{0}}{\pgfqpointxy{\xmincoord}{0}}}{\pgfqpoint{-0.5\yafaxiswidth}{\yafaxispad}}}
	\pgfpathlineto{\pgfpointadd{\pgfpointadd{\pgfplotspointrelaxisxy{0}{0}}{\pgfqpointxy{\xmaxcoord}{0}}}{\pgfqpoint{0.5\yafaxiswidth}{\yafaxispad}}}
	\pgfusepath{stroke}

}
\newcommand{\yafdrawyaxis}[2]{
	\pgfplotstransformcoordinatey{#1}\let\ymincoord=\pgfmathresult 
	\pgfplotstransformcoordinatey{#2}\let\ymaxcoord=\pgfmathresult 
	\pgfsetlinewidth{\yafaxiswidth} 
	\pgfsetcolor{yafaxiscolor}
	\pgfpathmoveto{\pgfpointadd{\pgfpointadd{\pgfplotspointrelaxisxy{0}{0}}{\pgfqpointxy{0}{\ymincoord}}}{\pgfqpoint{\yafaxispad}{-0.5\yafaxiswidth}}}
	\pgfpathlineto{\pgfpointadd{\pgfpointadd{\pgfplotspointrelaxisxy{0}{0}}{\pgfqpointxy{0}{\ymaxcoord}}}{\pgfqpoint{\yafaxispad}{0.5\yafaxiswidth}}}
	\pgfusepath{stroke}
}

\pgfplotscreateplotcyclelist{yaf}{% 
{yafcolor1,mark options={scale=0.75},mark=o}, 
{yafcolor2,mark options={scale=0.75},mark=square},
{yafcolor3,mark options={scale=0.75},mark=triangle},
{yafcolor4,mark options={scale=0.75},mark=o},
{yafcolor5,mark options={scale=0.75},mark=o},
{yafcolor6,mark options={scale=0.75},mark=o},
{yafcolor7,mark options={scale=0.75},mark=o},
{yafcolor8,mark options={scale=0.75},mark=o}} 

\pgfplotsset{axis y line=left, axis x line=bottom,
	tick align=outside,
	compat = 1.3,
	tickwidth=\yafticklen,
	clip = false,
    x axis line style= {-, line width = 0pt, opacity = 0},
    y axis line style= {-, line width = 0pt, opacity = 0},
    x tick style= {line width = \yafaxiswidth, color=yafaxiscolor, yshift = \yafaxispad},
    y tick style= {line width = \yafaxiswidth, color=yafaxiscolor, xshift = \yafaxispad},
    x tick label style = {font=\scriptsize, yshift = \yaftlpad},
    y tick label style = {font=\scriptsize, xshift = \yaftlpad},
    every axis y label/.style = {at = {(ticklabel cs:0.5)}, rotate=90, anchor=center, font=\scriptsize, yshift = -\yaflabelpad},
    every axis x label/.style = {at = {(ticklabel cs:0.5)}, anchor=center, font=\scriptsize, yshift = \yaflabelpad},
    x tick label style = {font=\scriptsize, yshift = 1pt},
    grid = major,
    major grid style  = {dash pattern = on 1pt off 3 pt},
	every axis plot post/.append style= {line width=\yafaxiswidth} ,
	legend cell align = left,
	legend style = {inner sep = 1pt, cells = {font=\scriptsize}},
	legend image code/.code={%
		\draw[mark repeat=2,mark phase=2,#1] 
		plot coordinates { (0cm,0cm) (0.15cm,0cm) (0.3cm,0cm) };% 
	} 
}

\xspaceaddexceptions{$}

\newif\ifapx
\apxtrue
\begin{document}

\title{Discovering Descriptive Tile Trees}
\subtitle{by Mining Optimal Geometric Subtiles}

\author{Nikolaj Tatti \and Jilles Vreeken}
\institute{Advanced Database Research and Modeling\\
Universiteit Antwerpen\\
\url{{nikolaj.tatti,jilles.vreeken}@ua.ac.be}}

\maketitle

\begin{abstract}
When analysing binary data, the ease at which one can interpret results is very important. Many existing methods, however, discover either models that are difficult to read, or return so many results interpretation becomes impossible. 
Here, we study a fully automated approach for mining easily interpretable models for binary data. We model data hierarchically with noisy tiles---rectangles with significantly different density than their parent tile. To identify good trees, we employ the Minimum Description Length principle. 

We propose \Mondrian, a greedy any-time algorithm for mining good tile trees from binary data. Iteratively, it finds the locally \emph{optimal} addition to the current tree, allowing overlap with tiles of the same parent. A major result of this paper is that we find the optimal tile in only $\Theta(\numrows\numcols\min(\numrows,\numcols))$ time. \Mondrian can either be employed as a top-$k$ miner, or by MDL we can identify the tree that describes the data best.

Experiments show we find succinct models that accurately summarise the data, and, by their hierarchical property are easily interpretable.

\end{abstract}

\section{Introduction}
\label{sec:intro}

When exploratively analysing a large binary dataset, being able to easily interpret the results is of utmost importance. Many data analysis methods, however, have trouble meeting this requirement. With frequent pattern mining, for example, we typically find overly many and highly redundant results, hindering interpretation~\cite{han:07:freqpat}. Pattern set mining~\cite{bringmann:07:chosen,vreeken:11:krimp,debie:11:dami} tackles these problems, and instead provides small and high-quality sets of patterns. However, as these methods generally exploit complex statistical dependencies between patterns, the resulting models are often difficult to fully comprehend.

When analysing 0--1 data, the encompassing question is `how are the 1s distributed?'. In this paper, we focus on the underlying questions of `where are the ones?' and `where are the zeroes?'\!. To answer these questions in an easily interpretable manner, we propose to model the data hierarchically, by identifying trees of tiles, i.e.\ sub-matrices that are surprisingly dense or sparse compared to their parent tile. 
As an example, consider Figure~\ref{fig:toyexample}, in which we show a toy example of a hierarchical tiling, and the corresponding tile tree. As the figure shows, tiles model parts of the data, and subtiles provide refinements over their parents. Next, as an example on real data, consider Figure~\ref{fig:paleo-example}, in which we show the tiling our algorithm discovered on paleontological data. Very easily read, using only $14$ tiles, the tiling shows which regions of the data are relatively dense (dark), as well as where relatively few $1$s are found (light).

\tikzstyle{tile} = [rounded corners = 2pt, inner sep = 0pt, fill opacity = 1, anchor = south west]
\tikzstyle{treetile} = [rounded corners = 2pt, inner sep = 4pt, fill opacity = 1]

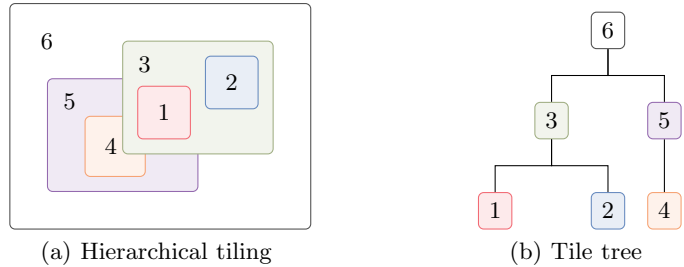
\begin{figure}[t!]
\begin{center}
\subfigure[Hierarchical tiling\label{fig:toy:data}]{
\begin{tikzpicture}
\node[tile, draw=black!70, fill=white, minimum height = 3cm, minimum width = 4cm]  at (0, 0) {};
\node[tile, draw=yafcolor1!70, fill=yafcolor1!10, minimum height = 1.5cm, minimum width = 2cm] at (0.5, 0.5) {};
\node[tile, draw=yafcolor2!70, fill=yafcolor2!10, minimum height = 0.8cm, minimum width = 0.8cm, align = left] at (1, 0.7) {$4$\ };
\node[tile, draw=yafcolor3!70, fill=yafcolor3!10, minimum height = 1.5cm, minimum width = 2cm] at (1.5, 1) {};
\node[tile, draw=yafcolor4!70, fill=yafcolor4!10, minimum height = 0.7cm, minimum width = 0.7cm] at (1.7, 1.2) {$1$};
\node[tile, draw=yafcolor5!70, fill=yafcolor5!10, minimum height = 0.7cm, minimum width = 0.7cm] at (2.6, 1.6) {$2$};
\node  at (0.5, 2.5) {$6$};
\node  at (0.8, 1.7) {$5$};
\node  at (1.8, 2.2) {$3$};
\end{tikzpicture}}\hspace{2cm}
\subfigure[Tile tree\label{fig:toy:tree}]{
\begin{tikzpicture}[level distance = 1.2cm]
	\node [treetile, draw=black!70, fill=white] {6}
		[edge from parent fork down] 
		child { node [treetile, draw=yafcolor3!70, fill=yafcolor3!10] {3}
			child { node [treetile, draw=yafcolor4!70, fill=yafcolor4!10] {1} }
			child { node [treetile, draw=yafcolor5!70, fill=yafcolor5!10] {2} }
		}
		child { node [treetile, draw=yafcolor1!70, fill=yafcolor1!10] {5}
			child { node [treetile, draw=yafcolor2!70, fill=yafcolor2!10] {4} }
		};
\end{tikzpicture}}
\end{center}
\vspace{-1.5em}
\caption{Toy example of a tiled database, and the corresponding tile tree structure}\label{fig:toyexample}
\end{figure}

\begin{figure}[b!]
\begin{center}%
\includegraphics[width=0.7\textwidth]{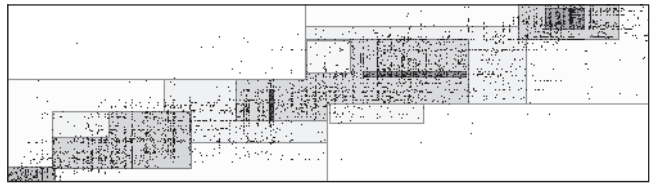}%
\end{center}%
\vspace{-1em}
\caption{Tiling of the {\em Paleo} dataset. See Fig.~\ref{fig:paleo-clean} for a  cleaned version without $1$s}\label{fig:paleo-example}
\end{figure}

Clearly, we aim to mine descriptions that are succinct, non-redundant, and neither overly complex nor simplistic. We therefore formalise the problem in terms of the Minimum Description Length (MDL) principle~\cite{grunwald:07:book}, by which we can automatically identify the model that best describes the data, without having to set any parameters. For mining good models, we introduce \Mondrian, a heuristic any-time algorithm that iteratively greedily finds the optimal subtile and adds it to the current tiling. A major result of this paper is that we show that we can find such optimal subtiles in only $\Theta(\numrows\numcols\min(\numrows,\numcols))$, as opposed to $\Theta(\numrows^2\numcols^2)$ when done naively~\cite{gionis:04:geometric}. 

We are not the first to study the problem of hierarchical tiling. The problem was first introduced by Gionis et al.~\cite{gionis:04:geometric}, whom proposed a randomised approach as an alternative to the naive approach. Our \FindTile procedure, on the other hand, is deterministic and identifies \emph{optimal} subtiles. Moreover, our MDL formalisation requires no scaling parameters, making the method parameter-free. 

These differences aside, both methods assume an order on the rows and columns of the data; as for such data, a subtile can be straightforwardly defined by a `from' and `to' selection query. As such, we exploit that the data is ordered, as this allows us to generate more easily understandable and easily visually representable models for the data. Although many datasets naturally exhibit such order, e.g. spatially and/or temporally, not all data does. For unordered data, e.g. through spectral ordering, good orders can be discovered~\cite{gionis:04:geometric, fortelius:06:spectral, tatti:11:order}. 

Experimentation on our method shows we discover easily interpretable models that describe the data very well. \Mondrian mines trees that summarise the data succinctly, with non-redundant tile trees that consist of relatively few tiles. 

The paper is organised as follows. Section~\ref{sec:encoding} discusses preliminaries. Section~\ref{sec:search} gives the \Mondrian algorithm for mining tile trees, and Section~\ref{sec:optimal} details mining optimal subtiles. We discuss related work in Section~\ref{sec:related}, and experiment in Section~\ref{sec:exps}. We round up with discussion and conclusions. Due to lack of space, we give the proofs for Propositions~\ref{prop:gain}--\ref{prop:inequality} in  
\ifapx 
Appendix~\ref{sec:apx}.
\else 
the Appendix.\!\footnote{http://adrem.ua.ac.be/stijl/}
\fi

\section{Encoding Data with Tile Trees}
\label{sec:encoding}

We begin by giving the basic definitions we use throughout the paper, after which we discuss how we can measure the quality of a hierarchical tile set.

\vspace{0.2em}
\textbf{Notation}
%\note{Def: Dataset}
A \emph{binary dataset} $D$ is a binary matrix of size \by{\numrows}{\numcols} consisting
of $\numrows$ rows, binary vectors of size $\numcols$. We denote $(i, j)^{\textrm{th}}$ entry of $D$ by $D(i, j)$. We assume that both rows and columns have an order and, for simplicity, we assume that the indexing corresponds to the orders.

%\note{Def: Tile}
%\note{Def: Subtile}

A geometric \emph{tile} $X = (a, b) \times (c, d)$, where $1 \leq a \leq b \leq \numrows$ and $1
\leq c \leq d \leq \numcols$, identifies a consecutive submatrix of $D$. 
In contrast, for combinatorial tiles, the rows and columns are not required to be consecutive. In this paper, we focus on geometric tiles.
We say that $X_1 = (a_1, b_1) \times (c_1, d_1)$ is a subtile of $X_2 = (a_2, b_2)
\times (c_2, d_2)$ if $X_1$ is completely covered by $X_2$, that is, $a_2 \leq
a_1$, $b_1 \leq b_2$, $c_2 \leq c_1$, and $d_1 \leq d_2$. We will write $(i, j)
\in X$ if $a \leq i \leq b$ and $c \leq j \leq d$. 

%\note{Def: Tile Tree}
A \emph{tile tree} $\tree{T}$ is a tree of tiles such that each child of a tile
$X \in \tree{T}$ is a subtile of $X$. We will denote the children of $X$ by
$\children{X}$. In our setting, the order of the children matters, so we assume
that $\children{X}$ is a \emph{list} of tiles and not a set. We also assume
that the root tile always covers the whole data. Given a tile tree $\tree{T}$,
a tile $X \in \tree{T}$ and a subtile $Y$ of $X$, we will write $\tree{T} + X \to Y$
to denote a tile tree obtained by adding $Y$ as a last child of $X$. 

%\note{Def: Tile order}
Our next step is to define which data entries are covered by which tile.
Since we allow child tiles to overlap, the definition is involved---although intuition is simple: the first most-specific tile that can encode a cell, encodes its value, and all other tiles ignore it. More formally, given a tile tree $\tree{T}$, consider a post-order, that is, an order where the
child tiles appear before their parents and such that if $\children{X} =
\enpr{Y_1}{Y_L}$, then $Y_i$ is appears before $Y_{i + 1}$. Let $X \in
\tree{T}$. 
We define $\tid{X ; \tree{T}}$ to be the position of $X$ in the post-order. When $\tree{T}$ is clear from the context we will simply write $\tid{X}$.
An example of the post-order is given in Figure~\ref{fig:toy:tree}.
Using this order we can define which entries belong to which tile.
We define
%\note{Def: positives negatives of tile}
\[
	\cells{X; \tree{T}} = \set{(i, j) \in X \mid \text{there is no } Y \text{ with } (i, j) \in Y,  \tid{Y} < \tid{X} },
\]
that is, entries are assigned to the cells first-come first-serve, see Figure~\ref{fig:toy:data} as an example.
Among these entries, we define the number of $1$s and $0$s as
\[
\begin{split}
	\ones{X ; \tree{T}, D} = \abs{\set{(i, j) \in \cells{X ; \tree{T}} \mid D(i, j) = 1}}& \;\; \text{ and} \\
	\zeroes{X ; \tree{T}, D} = \abs{\set{(i, j) \in \cells{X ; \tree{T}} \mid D(i, j) = 0}}&\quad .
\end{split}
\]
Let us denote by $|\tree{T}|$ the number of tiles in a tree $\tree{T}$, i.e. $|\tree{T}|=|\{X \in \tree{T}\}|$, and denote by $\tree{T}_0$ the most simple tile tree consisting of only a root tile.

%\note{Def: Encoding a data with tile}

\vspace{0.2em}
\textbf{MDL for Tile Trees}
Our main goal is to find tile trees that summarise the data well; they should be  succinct yet highly informative on where the $1$s on the data are. 
We can formalise this intuition through the Minimum Description Length (MDL) principle~\cite{grunwald:07:book}, a practical version of Kolmogorov Complexity~\cite{vitanyi:93:book}. Both embrace the slogan {\em Induction by Compression}. The MDL principle can be roughly described as follows:
Given a dataset $\DB$ and a set of models $\Models$ for $\DB$, the best model $\Model \in \Models$ is the one that minimises
$L(\Model) + L(\DB \mid \Model)$
in which
$L(\Model)$ is the length, in bits, of the description of the model $\Model$, and
$L(\DB \mid \Model)$ is the length, in bits, of the data as described using $\Model$.

This is called two-part MDL, or {\em crude} MDL. This stands opposed to
{\em refined} MDL, where model and data are encoded together \cite{grunwald:07:book}.
We use two-part MDL because we are specifically interested in the model: the tile tree $\tree{T}^*$ that yields the minimal description length. Further, although refined MDL has stronger theoretical foundations, it cannot be computed except for some special cases~\cite{grunwald:07:book}. 
Before we can use MDL to identify good models, we will have to define how to encode a database given a tile tree, as well as how to encode a tile tree.

We encode the values of $\cells{X; \tree{T}}$ using prefix codes. The length of an optimal prefix code is given by Shannon entropy, i.e.\ $-\log P(\cdot)$, where $P(\cdot)$ is the probability of a value~\cite{cover:06:elements}. We have the optimal encoded length for all entries $\cells{X; \tree{T}}$ of a tile $X$ in a tile tree $\tree{T}$ as 
\[
	\lc{D \mid X, \tree{T}} = \ent{\ones{X; \tree{T}, D}, \zeroes{X; \tree{T}, D}},
\]
where $\ent{p, n} = - p \log \frac{p}{p+n} - n \log \frac{n}{p+n}$ is the scaled entropy.

%\note{Def: Cost of a tile}
In order to compare fairly between models, MDL requires the encoding to be lossless. Hence, besides the data, we also have to encode the tile tree itself.

We encode tile trees node per node, in reverse order, and add extra bits between the tiles to indicate the tree structure. We use a bit of value $1$ to indicate that the next tile is a child of the current tile, and $0$ to indicate that we have processed all child tiles of the current tile.
%\!\footnote{Although constant costs are slightly inefficient, this allows us to speed up the search.}
For example, the tree given in Figure~\ref{fig:toy:tree} is encoded, with $<$tile $ $$\mathit{i}$$>$ indicating an encoded tile, as
$$
\text{$<$tile $6$$>$}1
\text{$<$tile $5$$>$}1
\text{$<$tile $4$$>$}001
\text{$<$tile $3$$>$}1
\text{$<$tile $2$$>$}01
\text{$<$tile $1$$>$}000 \quad .
$$

To encode an individual tile, we proceed as follows. Let $X$ be a non-root tile and let
$Z = (a, b) \times (c, d)$ be the direct parent tile of $X$. As we know that $X$ is a subtile of $Z$, we know the end points for defining the area of $X$ fall within those of $Z$. As such, to encode the $4$ end points of $X$ we need only $4 \log(b - a + 1) + 4\log(d - c + 1)$ bits.

We also know that number of $1$s in $X$ are bounded by the area of $Z$, $(b - a + 1)(d - c + 1)$, and hence we can encode the number of $1$s in $X$ in $\log(b - a + 1) + \log(d - c + 1)$ bits. 
Note that although we can encode the number of $1$s more efficiently by using the geometry of $X$ instead of $Z$, this would introduce a bias to small tiles.

Next, to calculate the encoded size of a tile, we need to take the two bits for encoding the tree structure of $X$ into account. As describe above, one bit is used to indicate that $X$ has no more children and the other to indicate that $X$ is a child of $Z$. 
Putting this together, the encoded length of a non-root tile $X$ is
\[
	\lc{X \mid \tree{T}} = 1 + 1 + 5 \log(b - a + 1) + 5\log(d - c + 1)\quad.
\]
Let us now assume that $X$ is the root tile. Since we require that a root tile covers the whole data, we need to encode the dimensions of the data set, the number of $1$s in $X$, and following $1$ bit to indicate that all tiles have been processed. Unlike for the other tiles in the tree, we have no  upper bound for the dimensions of $X$, and therefore would have to use a so-called Universal Code~\cite{grunwald:07:book} to encode the dimensions---after which we could subsequently encode the number of $1$s in $X$ in $\log \numrows + \log \numcols$ bits. 
However, as the lengths of these codes are constant over all models for $D$, and  we can safely ignore them when selecting between models. 
For simplicity, for a root tile $X$ we therefore define $\lc{X \mid \tree{T}} = 0$.
%\note{Def: full cost}

As such, we have for the total encoded size of a database $D$ and a tile tree $\tree{T}$ \[
	\lc{D, \tree{T}} = \sum_{X \in \tree{T}} \lc{X \mid \tree{T}} + \lc{D \mid X, \tree{T} } \quad ,
\]
by which we now have a formal MDL score for tile trees.

\section{Mining Good Tile Trees}
\label{sec:search}
Now that we have defined how to encode data with a tile tree, our next step is to find the best tile tree, i.e. the tile tree minimising the total encoded length. That is, we want to solve the following problem.

%\note{Problem: finding optimal tree}
\begin{problem}[Minimal Tile Tree]
Given a binary dataset $D$, find a tile tree $\tree{T}$ such that the total encoded size, $\lc{D, \tree{T}}$, is minimised.
\end{problem}

%\note{Alg: Greedy Heuristic}
As simply as it is stated, this minimisation problem is rather difficult to solve. Besides that the search space of all possible tile trees is rather vast, the total encoded size $\lc{D, \tree{T}}$ does not exhibit trivial structure that we can exploit for fast search, e.g.
(weak) monotonicity.
Hence, we resort to heuristics. 

For finding an approximate solution to the Minimal Tile Tree problem, we propose the \Mondrian algorithm.\!%\footnote{named after the minimalistic-art movement De Stijl, to the art of which our models show resemblance.}
\footnote{named after the art movement De Stijl, to which art our models show resemblance.}
We give the pseudo-code as Algorithm~\ref{alg:mondrian}. We iteratively find that subtile $Y$ of a tile $X \in \tree{T}$ by which the total encoded size is minimised. We therefore refer to $Y$ as the optimal subtile of $X$. After identifying the optimal subtile, \Mondrian adds $Y$ into the tile tree, and continues inductively until no improvement can be made.%Note that it does not matter whether we use breath-first or depth-first strategy for growing the tree.
%the order at which we add optimal subtiles does not influence the outcome.

Alternative to this approach, we can also approximate the optimal $k$-tile tree. To do so, we adapt the algorithm to find the subtile $Y$ over all parent tiles $X \in \tree{T}$ that minimises the score---as opposed to our standard depth-first strategy. Note that by the observation above, for the $k$ at which the score is minimised, both strategies find the same tree.
%Due to type of structure in the data it discovers---nested tiles of varying density---the resulting models show likeness to the art of the famous painter Mondrian.

\begin{algorithm}[tb!]
\caption{\Mondrian$(D, \tree{T}, X)$}
\label{alg:mondrian}
\Input{dataset $D$, current tile tree $\tree{T}$, parent tile $X$}
\Output{updated tile tree $\tree{T}$}

$Y \define $ subtile of $X$ minimising $\lc{D, \tree{T} + X \to Y}$\;

\While {$\lc{D, \tree{T} + X \to Y} < \lc{D, \tree{T}}$} {
	$\tree{T} \define \Mondrian(D, \tree{T} + X \to Y, Y)$\;
	$Y \define $ subtile of $X$ minimising $\lc{D, \tree{T} + X \to Y}$\;
}
\Return{$\tree{T}$}\;
\end{algorithm}

%\note{Problem: finding a optimal subtile}
By employing a greedy heuristic, we have reduced the problem of finding the optimal tile tree into a problem of finding the optimal subtile.

\begin{problem}[Minimal Subtile]
\label{prob:subtile}
Given a binary dataset $D$, a tile tree $\tree{T}$, and a tile $X \in \tree{T}$,
find a tile $Y$ such that $Y$ is a subtile of $X$, and $\tree{T} + X \to Y$ is minimised.
\end{problem}

The main part of this paper details how to find an optimal subtile, a procedure  we subsequently use in \Mondrian.

\section{Finding the Optimal Tile}
\label{sec:optimal}
In this section we focus on finding the optimal subtile. Naively, we solve this by simply testing every possible subtile, requiring $\Theta(\numrows^2\numcols^2)$ tests, where $\numrows$ and $\numcols$ are the number of rows and columns in the parent tile, respectively~\cite{gionis:04:geometric}.

In this section, we present an algorithm that can find the optimal subtile in
$\Theta(\numrows^2\numcols)$.  In order to do that, we will break the problem into two
subproblems. The first problem is that for two \emph{fixed} integers $c \leq
d$, we need to find two integers $a \leq b$ such that the tile $(a, b) \times
(c, d)$ is optimal. Once we have solved this, we can proceed to find the optimal
tile by finding the optimal $(c, d)$.

We begin by giving an easier formulation of the function we want to optimise.
In order to do so, note that adding a subtile $Y$ to $X$ changes only
$\cells{T; \tree{X}}$, and does not influence (the encoded length of) other tiles in the tree. Hence, we expect to be able to express the difference in total encoded length between $\tree{T} + X \to Y$ and $\tree{T}$ in simple terms. In fact, we have the following theorem.

\begin{proposition}
\label{prop:gain}
Let $\tree{T}$ be a tile tree. Let $X \in \tree{T}$ be a tile and let $Y$ be a
subtile of $X$.  Define $\tree{T}' = \tree{T} + X \to Y$ and have $u = \ones{Y;
\tree{T}'}$,  $v = \zeroes{Y; \tree{T}'}$, $o = \ones{X; \tree{T}}$, and  $z =
\zeroes{X; \tree{T}}$. Then
\[
	\lc{D, \tree{T}'} - \lc{D, \tree{T}} = \ent{u, v} + \ent{o - u, z - v} - \ent{o, z} + \lc{Y \mid \tree{T}'}\quad.
\]
\end{proposition}

In order to find the optimal subtile it is enough to create an algorithm for finding an optimal subtile more dense than its parent tile. To see this, note that we can find the optimal tile by first finding the optimal \emph{dense} tile, and then find the optimal \emph{sparse} tile by applying the same algorithm on the $0$--$1$ inverse of the data. Once we have both the optimal optimal dense and optimal sparse tiles, we can choose the overall optimal subtile by MDL.

Let $X = (s, e) \times (x, y)$ be a tile, and $\tree{T}$ a tile tree with $X \in \tree{T}$. 
Assume that we are given indices $c$ and $d$.
Our goal in this section is to find those indices $a$ and $b$ such that $Y = (a, b) \times (c, d)$ is an optimal subtile of $X$.

Define two vectors, $p$ for positives and $n$ for negatives, each of length $e - s + 1$, that contain the number of $1$s and $0$s respectively, within the $i$th row of $X$,  $p_i = \abs{\set{(i + s - 1, w) \in \cells{X} \mid c \leq w \leq d, D(i + s - 1, w) = 1}}$, 
and $n_i = \abs{\set{(i + s - 1, w) \in \cells{X} \mid c \leq w \leq d, D(i + s - 1, w) = 0}}$. 

This allows us to define $\cnt{a, b; p} = \sum_{i = a}^{b} p_i$ (and similarly
for $n$).  Let $u = \cnt{a, b; p}$ and $v = \cnt{a, b; n}$.  It follows that
$\ones{Y; \tree{T} + X \to Y} = u$ and $\zeroes{Y; \tree{T} + X \to Y} = v$,
where $Y = (a, b) \times (c, d)$; those are the entries of $X$ now to be
encoded by $Y$. Let us define $\cost{a, b; p, n, o, z} = \ent{u, v} + \ent{o -
u, z - v}$. We will write $\cost{a, b}$, when $p$, $n$, $o$, $z$ are known from
the context. Proposition~\ref{prop:gain} states that minimising
$\lc{D, \tree{T}'}$ is equivalent to minimising $\cost{a, b}$.

Further, let us define $\freq{a, b; p, n} = \cnt{a, b; p} / (\cnt{a, b; p} +
\cnt{a, b; n})$ to be the \emph{frequency} of $1$s within $Y$.
Proposition~\ref{prop:gain} then allows us to formulate the optimisation
problem as follows.

\begin{problem}[Minimal Border Points]
\label{prob:border}
Let $p$ and $n$ be two integer vectors of the same length, $m$. Let $o$ and $z$
be two integers such that $\cnt{1, m; p} \leq o$ and $\cnt{1, m; n} \leq z$.
Find $1 \leq a \leq b \leq n$ such that $\freq{a, b; p, n} >  o / (o + z)$
and that $\cost{a, b}$ is minimised.
\end{problem}

The rest of the section is devoted to solving this optimisation problem.
Naively we could test every pair $(a, b)$, which however requires quadratic
time. Our approach is to ignore a large portion of suboptimal pairs, such
that our search becomes linear.

To this end,
let $p$ and $n$ be two vectors, and let $1 \leq b \leq \abs{p}$ be an
integer.  We say that $a \leq b$ is a \emph{head border} of $b$ if there are no
integers $i$ and $j$ such that $1 \leq i < a \leq j \leq b$ and $\freq{i, a - 1} \geq \freq{a, b}$.
Similarly, we say that $b \geq a$ is a \emph{tail border} of $a$ if there are
no indices $a \leq i \leq b < j \leq \abs{p}$ such that $\freq{i, b} \leq
\freq{b + 1, j}$.  We denote the list of all head borders by $\hborders{b, p, n}$
and the list of all tail borders by $\tborders{a, p, n}$.

Given a head border $a$ of $b$, we say that $a$ is a \emph{head candidate} if
there are no indices $1 \leq i < a \leq j \leq b$ such that $\freq{i, a - 1} \geq
\freq{j, b}$.
Similarly, we say that $b \in \tborders{a}$ is a \emph{tail candidate} of $a$
if there are no indices $a \leq i \leq b < j \leq \abs{p}$ such that $\freq{a, i}
\leq \freq{b + 1, j}$.
We denote the
list of all head candidates by $\hcands{b, p, n}$ and the list of all tail
candidates by $\tcands{a, p, n}$.

To avoid clutter, we do not write $p$ and $n$ wherever clear from context.

\begin{figure}[tb!]
\begin{center}
\subfigure[\label{fig:borders}]{
\begin{tikzpicture}
\node[tile, draw=yafcolor3, fill=yafcolor3!10, minimum height = 0.3cm, minimum width = 1cm, rounded corners = 0pt] at (0, 0) {};
\node[tile, draw=yafcolor3, fill=yafcolor3!30, minimum height = 1cm, minimum width = 1.3cm, rounded corners = 0pt] at (1, 0) {};
\node[tile, draw=yafcolor3, fill=yafcolor3!20, minimum height = 0.6cm, minimum width = 1cm, rounded corners = 0pt] at (2.3, 0) {};
\node[tile, draw=yafcolor3, fill=yafcolor3!60, minimum height = 2cm, minimum width = 1.2cm, rounded corners = 0pt] at (3.3, 0) {};
\node[tile, draw=yafcolor3, fill=yafcolor3!15, minimum height = 0.4cm, minimum width = 1.4cm, rounded corners = 0pt] at (4.5, 0) {};

\node[anchor = north] at (0, 0) {$a_1$};
\node[anchor = north] at (1, 0) {$a_2$};
\node[anchor = north] at (2.3, 0) {$a_3$};
\node[anchor = north] at (3.3, 0) {$a_4$};
\node[anchor = north] at (4.5, 0) {$a_5$};
\node[anchor = north] at (5.9, 0) {$a_6$};

\node[tile, draw=yafcolor3, fill=yafcolor3!10, minimum height = 2.5cm, minimum width = 0cm, rounded corners = 0pt] at (0, 0) {};
\node[anchor = south west, rotate = 90] at (0, 0.5) {frequency};

\end{tikzpicture}}\hspace{1cm}
\subfigure[\label{fig:inequality}]{
\begin{tikzpicture}
\node[tile, draw=black!70, fill=white, minimum height = 3cm, minimum width = 3.1cm]  at (0, 0) {}; 
\draw[dashed, draw=black!70] (0.8, 0) node [anchor = north] {$c$} -- (0.8, 3);
\draw[dashed, draw=black!70] (2.3, 0) node [anchor = north] {$d$} -- (2.3, 3);
\draw[dashed, draw=black!70] (0, 0.6) -- (3.1, 0.6) node [anchor = west] {$b$};
\draw[dashed, draw=black!70] (0, 1.2) -- (3.1, 1.2) node [anchor = west] {$a'$};
\draw[dashed, draw=black!70] (0, 1.8) -- (3.1, 1.8) node [anchor = west] {$i$};
\draw[dashed, draw=black!70] (0, 2.4) -- (3.1, 2.4) node [anchor = west] {$a$};
\node[tile, draw=yafcolor1!70, fill=yafcolor1!60, minimum height = 0.6cm, minimum width = 1.5cm] at (0.8, 0.6) {};
\node[tile, draw=yafcolor1!70, fill=yafcolor1!10, minimum height = 0.6cm, minimum width = 1.5cm] at (0.8, 1.2) {};
\node[tile, draw=yafcolor1!70, fill=yafcolor1!30, minimum height = 0.6cm, minimum width = 1.5cm] at (0.8, 1.8) {};
\end{tikzpicture}}
\end{center}
\vspace{-1.5em}
\caption{Example of how \FindTile considers head candidates. (a) $a_3$ is no
head border for $a_4 - 1$, as $\freq{a_2,a_3-1} > \freq{a_3,a_4-1}$. Although a
head border for $a_6 - 1$, $a_4$ is not a head candidate for $a_6 - 1$, as
$\freq{a_3,a_4-1}>\freq{a_5,a_6-1}$. (b) Proposition~\ref{prop:inequality} states that we can ignore $i$ as head
candidate for a tile to $j$, as by $\freq{u',i-1}>\freq{i,u-1}$ we know
$\freq{u',j-1}>\freq{i,j-1}$. 
}
\end{figure}
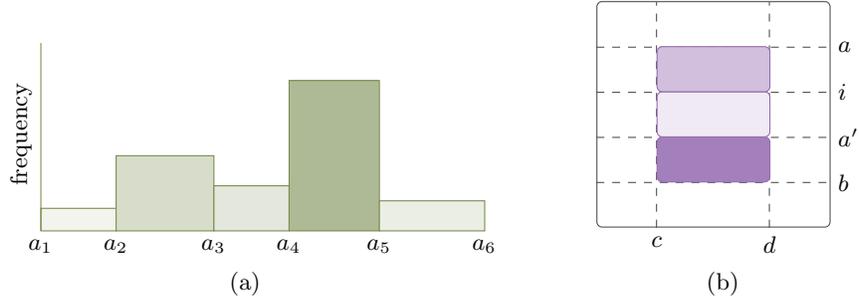

As an example, consider Figure~\ref{fig:borders}. Since $\freq{a_2, a_3 - 1} > \freq{a_3, a_4 - 1}$,
it follows that $a_3 \notin \hborders{a_4 - 1}$. Note that $a_4 \in \hborders{a_6 - 1}$ but 
$a_4 \notin \hcands{a_6 - 1}$ since $\freq{a_2, a_4 - 1} > \freq{a_5, a_6 - 1}$.

We are now ready to state the main result of this section: in order to find the optimal
tile we need to only study head and tail candidates.

\begin{proposition}
\label{prop:inequality}
Let $p$ and $n$ be two vectors and let $o$ and $z$ be two integers.
Let $i\leq j$ be two indices such that $\freq{i, j} > o / (o + z)$.
Then there are $a \leq b$ such that $\cost{a, b} \leq \cost{i, j}$, $a \in \hcands{b}$ and $b \in \tcands{a}$. 
\end{proposition}

Proposition~\ref{prop:inequality} is illustrated in
Figure~\ref{fig:inequality}.  Since $\freq{a, i - 1} \geq \freq{i, a'}$, we
know that $i \notin \hcands{b}$.  Proposition~\ref{prop:inequality} implies
that we can safely ignore $(i, b)$ and consider instead $(a', b)$ or $(a, b)$.

%\subsection{Discovering Borders}
Proposition~\ref{prop:inequality} states that we need to only study candidates, a subset of borders.
Fortunately, there exists an efficient algorithm to construct a border list
$\hborders{b + 1}$ given the existing list $\hborders{b}$~\cite{calders:07:mining}. The approach relies on several lemmata.

Let $\enpr{a_1}{a_K} = \hborders{b}$.  We claim that
$\hborders{b + 1} \subseteq \enpr{a_1}{a_K, b + 1}$.

\begin{lemma}
\label{lem:bordprevious}
If $a \leq b$ and  $a \notin \hborders{b}$, then $a \notin \hborders{b + 1}$.
\end{lemma}

\begin{proof}
By definition there are $i$ and $j$ such that 
$1 \leq i < a \leq j \leq b$ such that $\freq{i, a - 1} \geq \freq{a, j}$. These indices
are valid for $b + 1$, hence $a \notin \hborders{b + 1}$.
\qed\end{proof}

Hence, in order to construct $\hborders{b + 1}$, we only need to delete entries from $\enpr{a_1}{a_K, b + 1}$.
Let us define a \emph{head frequency} $\hfreq{b} = \max_{i \leq b} \freq{i, b}$
and a \emph{tail frequency} $\tfreq{a} = \max_{a \geq i} \freq{a, i}$.
The following two lemmata 
say that
the last entry in $\hborders{b + 1}$ has to be the smallest index $j$ such that
$\freq{j, b} = \hfreq{b}$, and that the borders of $b + 1$ smaller than $j$ are
all included in $\hborders{b}$.

\begin{lemma}
\label{lem:maximal}
Let $j$ be the smallest index s.t. $\freq{j, b} = \hfreq{b}$. Then
$j = \max \hborders{b}$.
Let $j$ be the largest index s.t. $\freq{a, j} = \tfreq{a}$. Then
$j = \min \tborders{a}$.
\end{lemma}

\begin{proof}
First note
that $j \in \hborders{b}$. Let $i$ be an index $j <  i \leq b$.
We have $\freq{i, b} \leq \freq{j, b}$ which implies that $\freq{j, i - 1} \geq \freq{i, b}$.
This implies that $i \notin \hborders{b}$. The case for $\tborders{a}$ is similar.
\qed\end{proof}

\begin{lemma}
\label{lem:stop}
Let $a \in \hborders{b}$. 
Let $k$ be the smallest index such that $\freq{k, b + 1} = \hfreq{b + 1}$.
If $a < k$, then $a \in \hborders{b + 1}$.
\end{lemma}

\begin{proof}
Assume that $a \notin \hborders{b + 1}$, that is, there are $i$ and $j$
such that $1 \leq i < a \leq j \leq b + 1$ such that $\freq{i, a - 1} \geq
\freq{a, j}$. We must have $j = b + 1$.  Otherwise $a \notin \hborders{b}$.
Note that $\freq{a, b + 1} < \freq{k, b + 1}$, which implies $\freq{a, k - 1} < \freq{a, b + 1}$.
Since $k - 1 \leq b$, we have $a \notin \hborders{b}$, which is a contradiction.
\qed\end{proof}

These lemmata give us a simple approach. Start from $\enpr{a_1}{a_K, b + 1}$
and delete entries until you find index $k$ such that $\freq{k, b + 1}$ is maximal.
We will see later in Proposition~\ref{prop:correct} that we can easily check the maximality.

Unfortunately, as demonstrated in~\cite{calders:07:mining} there can be $\Theta(b^{2/3})$
entries in $\hborders{b}$. Hence, checking every pair will not quite yield a linear algorithm. In order to achieve linearity, we use two additional
bounds. Consider Figure~\ref{fig:borders}. We have $\hborders{a_5 - 1} =
\pr{a_1, a_2, a_4}$.  First, since $\tfreq{a_5} = \freq{a_5, a_6 - 1} >
\freq{a_1, a_2 - 1}$, we have $a_5 - 1 \notin \tcands{a_1}$.  Consequently, we
do not need to check the pair $(a_1, a_5 - 1)$. Secondly, we know that for any
$k \geq a_5$ we have $\freq{a_5, k} \leq \tfreq{a_5} \leq \freq{a_3, a_4 - 1}$.
Hence, $a_4 \notin \hcands{k}$ and we can ignore $a_4$ after we have checked $(a_4, a_5 - 1)$.
We can now put these ideas together in a
single algorithm, given as Algorithm~\ref{alg:scan}, and which we will refer to as the \textsc{Scan} algorithm.

\begin{algorithm}[tb!]
\caption{\textsc{Scan}$(p, n, o, z)$}
\label{alg:scan}
\Input{integer vectors $p$ and $n$, number of 1s (0s) in the parent tile, $o$ ($z$) }
\Output{an interval $t$ solving Problem~\ref{prob:border}}
$\mathit{best} \define \infty$; $t \define (0, 0)$;
$B \define C \define \emptyset$\;

\ForEach{$b = 1, \ldots, \abs{p}$} {
	push $b$ to the front of $B$\;
	push $b$ to the front of $C$\;
	\While {$\abs{B} > 1$ \AND~$\freq{B_1, b; p, n} \leq \freq{B_2, B_1 - 1; p, n}$} {
		\nllabel{alg:find:clean}
		\lIf {$B_1 = C_1$} {remove $C_1$}
		remove $B_1$\;
	}
	\While {$\abs{C} > 1$ \AND~$\freq{C_2, C_1 - 1; p, n} \geq \tfreq{b + 1}$} {
		$c \define \cost{C_1, b}$\;
		\nllabel{alg:find:test1}
		\lIf {$c < \mathit{best}$} { 
			$t \define (C_1, b)$; $\mathit{best} \define c$
		}

		remove $C_1$\;
	}
	$c \define \cost{C_1, b}$\;
	\nllabel{alg:find:test2}
	\lIf {$c < \mathit{best}$} 
		{$t \define (C_1, b)$; $\mathit{best} \define c$}
	
}
\Return{$t$}\;

\end{algorithm}

Proposition~\ref{prop:inequality} stated that it is enough to consider intervals where then end points are each other candidates. The next proposition shows that \textsc{Scan} actually tests all such pairs. Consequently, we are guaranteed to find the optimal solution.

\begin{proposition}
\label{prop:correct}
Let $p$ and $n$ be count vectors and let $o$ and $z$ be two integers.
$\textsc{Scan}(p, n, o, z)$ tests every pair $(a, b)$ where $a \in \hcands{b}$ and $b \in  \tcands{a}$.
\end{proposition}

To show this, we first need the following lemma.
\begin{lemma}
\label{lem:bordmonotone}
Let $\enpr{a_1}{a_L} = \hborders{b}$. Then $\freq{a_{k - 1}, a_k - 1} < \freq{a_k, a_{k + 1} - 1}$.
\end{lemma}
\begin{proof}
Assume that $\freq{a_{k - 1}, a_k - 1} \geq \freq{a_k, a_{k + 1} - 1}$. Then $a_k \notin \hborders{b}$.
\qed\end{proof}

By which we can proceed with the proof for Proposition~\ref{prop:correct}.

\begin{proof}
Let us first prove that $B$ at $b$th step is equal to $\hborders{b}$. We prove
this using induction. The case $b = 1$ is trivial and assume that the result
hold for $b - 1$. Let $\enpr{a_1}{a_L} = \hborders{b - 1}$.
Lemma~\ref{lem:bordprevious} implies that $\hborders{b} \subseteq \enpr{a_1}{a_L, b}$.

Assume that $\freq{b, b} > \freq{a_L, i - 1} = \hfreq{b - 1}$.  By definition,
$b \in \hborders{b}$. Lemma~\ref{lem:maximal} implies that $b$ is the smallest
index $k$ for which $\freq{k, b} = \hfreq{b}$. Lemma~\ref{lem:stop} now states
that $\hborders{b} = \enpr{a_1}{a_L, b}$ which is exactly what we get since the
while loop on Line~\ref{alg:find:clean} is not executed.

Assume that $\freq{b, b} \leq \freq{a_L, i - 1}$. Then $b \notin \hborders{b}$
and indeed it is deleted in the first run of the while loop
(Line~\ref{alg:find:clean}).  Let $a_k \in \hborders{b}$ be the first entry in
$B$ after the while loop has finished. Let $a_l \in \hborders{b}$ be the smallest index
for which $\freq{a_l, b} = \hfreq{b}$.  We claim that $k = l$. If $l > k$,
then $\freq{a_l, b} \leq \freq{a_{l - 1}, a_{l} - 1}$ which implies that $\freq{a_l, b} \leq \freq{a_{l - 1}, b}$
which is a contradiction. Assume that $l < k$. By definition of $k$,
we have $\freq{a_{k - 1}, a_k - 1} < \freq{a_k, b}$.
Lemma~\ref{lem:bordmonotone} implies that $\freq{a_l, a_{k} - 1} \leq \freq{a_{k - 1}, a_k - 1}$.
Hence, $\freq{a_l, b} < \freq{a_k, b}$, which is a contradiction. 
Consequently, $k = l$. Lemma~\ref{lem:stop} now states
that $\hborders{b} = \enpr{a_1}{a_k}$ which is exactly what we have.

Now that we have proved that $B$ at $b$th step is equal to $\hborders{b}$.  Let
us consider the list $C$. Let $a \in B \setminus C$. This means that $a$ was
deleted during some previous round, say $k < i$, and that there is $j$ such
that $\freq{j, a - 1} \geq \tfreq{k + 1} \geq \freq{k + 1, b}$. Hence $a$ is
not a head candidate of $b$. Consequently, all head candidates of $b$ are
included in $C$ at $b$th step.

Not all entries of $C$ are tested during the $b$th step.  Assume that we
have completed $b$th step and $C_k$ is not tested ($k > 1$).  Since $C$ is a
subset of the border list, Lemma~\ref{lem:bordmonotone} implies that
$\freq{C_k, C_{k - 1} - 1} \leq \freq{C_2, C_1 - 1} < \tfreq{b + 1}$. 
There is $j$ such that $\tfreq{b + 1} = \freq{b + 1, j}$. This implies that $b$
is not a tail candidate for $C_k$, which completes the proof.
\qed\end{proof}

Let us finish this section by demonstrating the linear execution time
$\Theta(\abs{p})$ of \textsc{Scan}.  To this end, note that we have three while-loops in the algorithm: two inner and one outer. During each iteration of the
first inner loop we delete an entry from $B$, a unique number between $1$ and
$\abs{p}$.  Consequently, the \emph{total} number of times we execute the first inner
loop is $\abs{p}$, at maximum.  Similarly, for the second inner loop. The outer loop is executed $\abs{p}$ times. Next, note that there are two non-trivial subroutines in the algorithm.
First, on Lines 6 and 9, we need to compute frequencies. This can be done in constant time by, e.g. precomputing $c_j = \sum_{i = 1}^j p_i$, and then using the identity $\cnt{i, j; p} = c_j - c_{i - 1}$.  
Secondly, on Line 9, we need to compute $\tfreq{b}$. We can precompute this in linear time by using the algorithm given in~\cite{calders:07:mining}, which involves computing tail borders (equivalent to computing $B$ in \textsc{Scan}) and applying Lemma~\ref{lem:maximal}. This shows that the total execution time for the algorithm is $\Theta(\abs{p})$.

Now that we have a linear algorithm for discovering an optimal tile given a  fixed set of columns, we need an algorithm for discovering the columns themselves.
We employ a simple quadratic enumeration given in Algorithm~\ref{alg:findtile}.
Note that \textsc{Scan} assumes that the optimal tile is more dense than the background tile, we have to call \textsc{Scan} twice, once normally to find the optimal dense subtile, and once with ones and zeroes reversed to find the optimal sparse subtile.

\begin{algorithm}[tb!]
\caption{\FindTile$(X, \tree{T}, D)$}
\label{alg:findtile}
\Input{parent tile $X = (s, e) \times (x, y)$, current tile tree $\tree{T}$, dataset $D$}
\Output{$B$, a tile optimizing $\tree{T} + X \to B$, see Problem~\ref{prob:subtile}  }

$o \define \ones{X; \tree{T}}$;
$z \define \zeroes{X; \tree{T}}$;
$B \define X$\;

\ForEach {$c$ and $d$ such that $x \leq c \leq d \leq y$} {
	update $p$ and $n$\;

	$(a, b) \define \textsc{Scan}(p, n, o, z)$\; 
	$Y \define (a + s - 1, b + s - 1) \times (c, d)$\;
	\lIf {$\lc{\tree{T} + X \to Y, D} < \lc{\tree{T} + X \to B, D}$} {$B \define Y$}
	$(a, b) \define \textsc{Scan}(n, p, z, o)$\;
	$Y \define (a + s - 1, b + s - 1) \times (c, d)$\;
	\lIf {$\lc{\tree{T} + X \to Y, D} < \lc{\tree{T} + X \to B, D}$} {$B \define Y$}
}
\Return{$B$}\;
\end{algorithm}

The computational complexity of \FindTile is $\Theta(\numrows^2\numcols)$ where $\numrows$ is the
number of columns and $\numcols$ is the number of rows in the parent tile. 
However, if $\numcols$ is smaller than $\numrows$, we can transpose the parent tile and obtain a $\Theta(\numrows\numcols\min(\numrows, \numcols))$ execution time.

\section{Related Work}\label{sec:related}

Frequent itemset mining~\cite{agrawal:94:fast} is perhaps the most well-known example of pattern mining. Here, however, we are not just interested in the itemsets, but also explicitly want to know which rows they cover. Moreover, we are not interested in finding all tiles, but aim to find tilings that describe the data well.

Mining sets of patterns that describe a dataset is an actively studied topic~\cite{bringmann:07:chosen, tatti:08:decomposable, vreeken:11:krimp}. Related in that it employs MDL, is the \textsc{Krimp} algorithm~\cite{vreeken:11:krimp}, which proposed the use of MDL to identify pattern sets. Geerts et al.\ discuss mining large tiles of only $1$s~\cite{geerts:04:tiling}. Different from these approaches, our models are hierarchical, and do allow for noise within tiles.

Kontonasios and De Bie discuss ranking a candidate collection of tiles, employing a maximum entropy model of the data to measure the interestingness of a tile~\cite{konto:10:sdm,debie:11:dami}. 
Boolean matrix factorisation~\cite{miettinen:08:discrete} can be regarded as a tile mining. The goal is to find a set of Boolean factors such that the Boolean product thereof (essentially tiles of only $1$s) approximates the dataset with little error. Similarly, bi-clustering can be regarded as a form of tiling, as it partitions the rows and columns of a dataset into rectangles~\cite{pensa:05:bicluster}.
Compared to these approaches, a major difference is that we focus on easily inspectable hierarchical models, allowing nested refinements within tiles.

Most closely related to \Mondrian is the approach by Gionis et
al.~\cite{gionis:04:geometric}, who proposed mining hierarchies of tiles, and
gave a randomised heuristic for finding good subtiles. We improve over this
approach by formally defining the problem in terms of MDL, employing a richer
modelling language in the sense that it allows tiles with the same parent to overlap,
introducing a  deterministic iterative any-time algorithm, that given a tile
tree efficiently finds the optimal subtile. 
Since the approach by Gionis et al.~\cite{gionis:04:geometric} does not use
MDL as a stopping criterion, it is not possible to compare both methods directly. In principle, it is possible to adopt their search strategy to our score. A fair comparison between the two search strategies, however, is not trivial since the randomised search depends on a parameter, namely the number of restarts. This parameter acts as a trade-off between the expected performance and execution time. Choosing
this parameter is difficult since there are no known bounds for the expected
performance.

Our approach for discovering optimal subtiles is greatly inspired by the work
of Calders et al.~\cite{calders:07:mining} in which the goal was to compute the head
frequency, $\hfreq{i}$, given a stream of binary vectors.

As \Mondrian is an iterative any-time algorithm, and hence iterative data mining approaches are related. The key idea of these approaches is to iteratively find the result providing the most novel information about the data with respect to what we already know~\cite{hanhijarvi:09:tell, debie:11:dami, mampaey:11:tell}. Here, we focus on hierarchical tiles, and efficiently find the locally optimal addition.

\section{Experiments}\label{sec:exps}

In this section we empirically evaluate our approach. 
We implemented our algorithms in C++, and provide the source code, along with the synthetic data generator.\!\footnote{\url{http://adrem.ua.ac.be/stijl/}}
All experiments were executed single-threaded on Linux machines with Intel Xeon X5650 processors (2.66GHz) and 12 GB of memory.

We use the shorthand notation $L\%$ to denote the compressed size of $\DB$ with the tile tree $\tree{T}$ as discovered by \Mondrian relative to the most simple tree $\tree{T}_0$, 
$\frac{L(\DB,\tree{T})}{L(\DB,\tree{T}_0)}\%$, 
wherever $\DB$ and $\tree{T}$ are clear from context.

We do not compare to the naive strategy of finding optimal subtiles as $\Theta(\numrows^2\numcols^2)$ execution time is impractical even for very small datasets.

\begin{table}[t!]
\caption{Results of \Mondrian on five datasets. Shown are, per dataset, number of rows and columns, overall density, and for resp. without and with overlap, the relative compression $L\%$ (lower is better), number of discovered tiles, and wall-clock runtime.}\label{tbl:data}
\begin{tabular*}{\textwidth}{@{\extracolsep{\fill}}l c@{ } rr r c@{ } rrr c@{\hspace{0.3em}} rrr}
\toprule
 && & & && \multicolumn{3}{l}{\textbf{Disjoint}} && \multicolumn{3}{l}{\textbf{Overlap}}\\
\cmidrule{7-9} \cmidrule{11-13}
\emph{Dataset} && $\numrows$ & $\numcols$ & $\%1$s &&
 $L\%$ & $|\tree{T}|$ & \emph{time} &&
 $L\%$ & $|\tree{T}|$ & \emph{time} \\
\midrule
Composition && 240 & 240 & $23.2$ &&
 $81.72$ & $8$ & $57$s && 
 $81.58$ & $7$ & $1$m$23$s\\[0.3em]
Abstracts && $859$ & 
$541$ & $6.6$ && 
 $89.59$ & $14$ & $16$m$03$s && 
 $89.54$ & $14$ & $27$m$54$s \\

DNA Amp. && $4\,590$ & $391$ & $1.5$ &&
 $61.91$ & $466$ & $334$m && 
 $61.61$ & $446$ & $625$m \\
Mammals && $2\,183$ & $121$ & $20.5$ &&
 $54.69$ & $55$ & $1$m$37$s && 
 $54.62$ & $50$ & $3$m$06$s \\
Paleo && $501$ & $139$ & $5.1$ &&
 $80.23$ & $14$ & $39$s && 
 $79.07$ & $13$ & $1$m$22$s\\
\bottomrule
\end{tabular*}
\end{table}

%\subsection{Datasets}
%\paragraph{Datasets}
\vspace{0.2em}\textbf{Datasets}
We evaluate our measure on one synthetic, and four publicly available real world datasets. The \by{240}{240} synthetic dataset {\em Composition} was generated to the likeness of the famous Mondrian painting `Composition II in Red, Blue, and Yellow', where we use different frequencies of $1$s for each of the colours.
{\em Abstracts} contains the abstracts of papers accepted at ICDM up to 2007, for which we take the words with a frequency of at least $0.02$ after stemming and removing stop words~\cite{debie:11:dami}.
The {\em DNA} amplification data contains DNA copy number amplifications. Such copies are known to activate oncogenes and are the hallmarks of nearly all advanced tumours~\cite{myllykangas:06:dna}. 
The {\em Mammals} presence data consists of presence records of European mammals\footnote{Available for research purposes: \url{http://www.european-mammals.org}} within geographical areas of $50 \times 50$ kilometers~\cite{mitchell-jones:99:atlas}.
Finally, {\em Paleo} contains information on fossil records\footnote{NOW public release 030717 available
from~\cite{fortelius:06:spectral}.}
found at specific palaeontological sites in Europe~\cite{fortelius:06:spectral}.

We give the basic properties of these datasets in Table~\ref{tbl:data}.
To obtain good orders for the real world datasets, we applied SVD, that is, we ordered items and transactions based on first left and right eigenvectors.

%\subsection{Synthetic Data}
%\paragraph{Synthetic Data} settings
\vspace{0.2em}\textbf{Synthetic Data}
As a sanity check, we first investigate whether \Mondrian can reconstruct the model for the {\em Composition} data. We ran experiments for both the disjoint and the overlapping tile settings. 
With the latter setup we perfectly capture the underlying model in only $7$ tiles. We show the data and discovered model as Figure~\ref{fig:composition}; each of the rectangles in the painting are represented correctly by a tile, including the crossing vertical and horizontal bars. When we require disjoint tiles, the fit of the model is equally good, however the model requires one additional tile to model the crossing bars.

\begin{figure}[t]
\begin{center}
\subfigure[Composition\label{fig:composition}]{%
\includegraphics[width=0.25\textwidth]{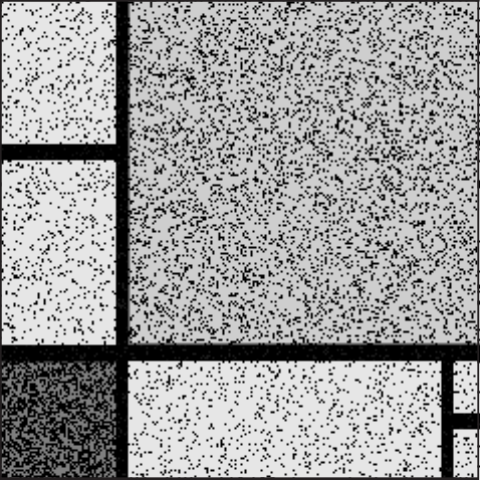}
}
\hspace{1em}
\subfigure[Paleo (transposed)\label{fig:paleo-clean}]{%
\parbox[r]{0.6\textwidth}{%
\vspace{-4.5em}%
\includegraphics[width=0.6\textwidth]{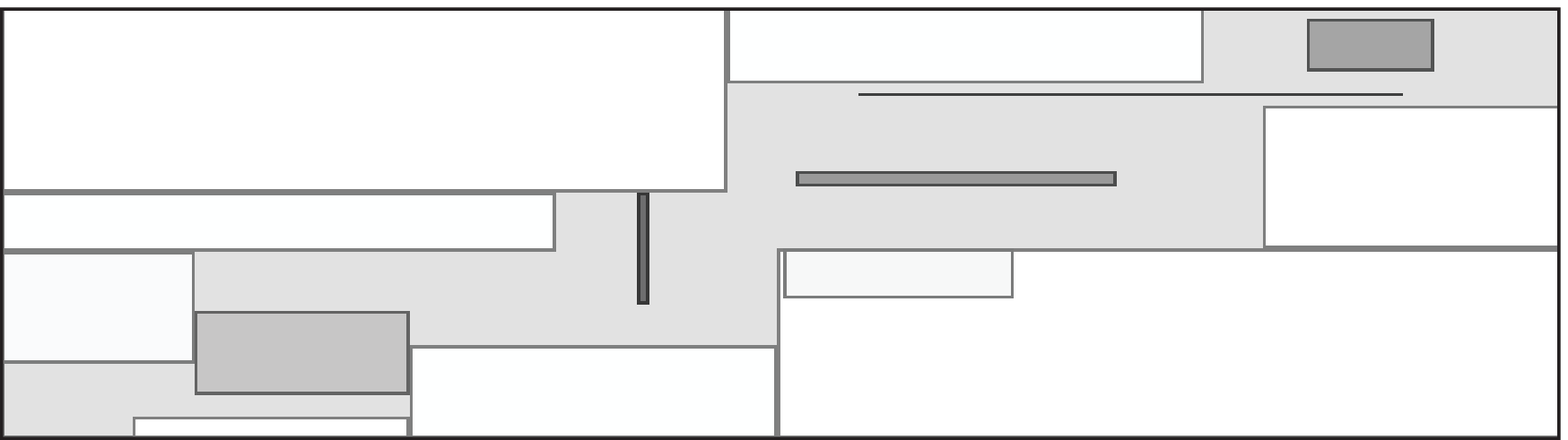}\\[1em]%
\vspace{1em}%
\includegraphics[width=0.6\textwidth]{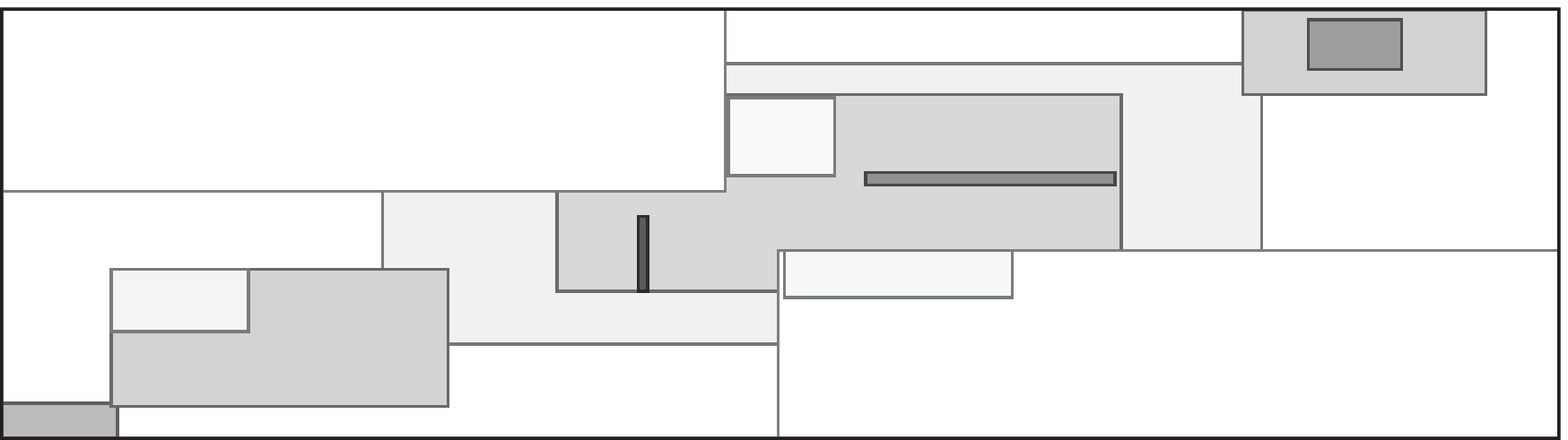}%
}}
\end{center}
\vspace{-1.5em}
\caption{Results of \Mondrian on (a) {\em Composition} and (b) {\em Paleo}, with (top) the disjoint hierarchical tiling, and (bottom) the tiling allowing overlap within the same parent tile. For {\em Paleo} we do not show individual $1$s. Darker tiles correspond to higher frequency}
\end{figure}

\vspace{0.2em}\textbf{Quantitative Analysis}
Next, we consider quantitative results on the real datasets. We run \Mondrian both for disjoint and overlapping tiles. Table~\ref{tbl:data} gives the relative compressed size $L\%$, the number of tiles in the returned tile trees, and the wall-clock time it took to find these models.

These results show \Mondrian finds trees that summarise the data well. The relative compressed sizes tell the data is described succinctly, while the tile trees remain small enough to be considered by hand; even for {\em DNA}, by the hierarchical nature of the model, the user can quickly read and understand the model.

For {\em Mammals} we see that whereas the baseline model $\tree{T}_0$ requires $193\,315$ bits, the \Mondrian model with overlapping tiles only requires $105\,589$ bits. Note that, as the total compressed size essentially is the log-likelihood of the model and the data, a gain of a single bit corresponds to a twice as likely model.

In all experiments, allowing overlap results in better models. Not only do they give more succinct data descriptions, the discovered tile trees are also simpler, requiring fewer tiles to do capture the structure of the data. 
By allowing overlap, the search space is expanded, and hence more computation is required: on average, in our experiments, twice as much. 

On these datasets, the current \Mondrian implementation requires from seconds up to a few hours of runtime. By its iterative any-time nature, users, however, can already start to explore models while in the background further refinements are calculated.

%\subsection{Real Data: Qualitative Analysis}
%\paragraph{Real Data: Qualitative Analysis}
\vspace{0.2em}\textbf{Qualitative Analysis} 
Next, we investigate the discovered models in more detail. To this end, we first use the {\em Paleo} data as by its modest size it is easily visually representable. In Figure~\ref{fig:paleo-clean} we show the result of \Mondrian on this data, with the top figure the result of allowing only disjoint tiles, and in the bottom figure when allowing overlap. Darker toned tiles correspond to more dense areas of the data. For clarity, we here do not show the individual $1$s (as we did in Fig.~\ref{fig:paleo-example}, which corresponds to the bottom plot of Fig.~\ref{fig:paleo-clean}).

The first thing we note, is that the two results are quite alike. The model with overlap, however, is a bit simpler and `cleaner': the relatively dense areas are of the data are easier to spot for this model, than for the disjoint one. Second, it uses the hierarchical property as intended: in the top right corner, for instance, we see a dense, dark-grey tile within a lighter tinted square, within a very sparse tile. While for reasons of space we can only show these examples, these are observations that hold in general---by which it may come at no surprise that by allowing overlap we obtain better MDL scores. 

Next, we inspect the results on {\em Abstracts}. This sparse dataset has no natural order by itself, and when we apply SVD to order it, we find most of the $1$s are located in the top-left corner of the data. When we apply \Mondrian, we see it correctly reconstructs this structure. Due to lack of space, however, we do not give the visual representation. Instead, we investigate the most dense tile, which covers the top-left corner. We find that it includes frequent words that are often used in conjunction in data mining abstracts, including {\em propose}, {\em efficient}, {\em method}, {\em mine}, and {\em algorithm}. Note that, by design, \Mondrian gives a high level view of the data; that is, it tells you where the ones are, not necessarily their associations. Extending it to recognise structure within tiles is future work.

\section{Discussion}\label{sec:discuss}

The experiments show \Mondrian discovers succinct tile trees that summarise the  data well. Importantly, the discovered tile trees consist of only few tiles, and are even easier inspected by the hierarchical property of our models.

The complexity of \Mondrian is much lower than that of the naive locally optimal approach; as with $\Theta(\numrows\numcols\min(\numrows,\numcols))$ its complexity is only squared in the smallest dimension of the data. However, for datasets with both many rows and columns, runtimes may be non-trivial. \Mondrian, however, does allow ample opportunity for optimisation. \FindTile, for instance, can be trivially run in parallel per parent tile, as well as over $a$ and $b$. 
%\note{something about bitmasks, data structures, blabla?}

As there is no such thing as a free lunch, we have to note that MDL is no magic wand. In the end, constructing an encoding involves choices---choices one can make in a principled manner (fewer bits is better), but choices nevertheless. Here, our choices were bounded by ensuring optimality of \FindTile. As such, we currently ignore globally optimal encoding solutions, such as achievable by maximum entropy modelling~\cite{debie:11:dami}. Although we could so obtain globally optimality of the encoding, the effects of adding a tile become highly unpredictable, which would break the locally optimal search of \FindTile.

We assume the rows and columns of the data to be ordered. That is, in the terminology of~\cite{gionis:04:geometric}, we are interested in geometric tiles. Although~\cite{gionis:04:geometric,fortelius:06:spectral} showed good geometric tilings can be found on spectrally ordered data, it would make for engaging research to investigate whether we can find good orderings on the fly, that is while we are tiling, ordering the data such that we optimise our score.

\section{Conclusion}\label{sec:concl}

We discussed finding good hierarchical tile-based models for binary data. We formalised the problem in terms of MDL, and introduced the \Mondrian algorithm for greedily approximating the score on binary data with ordered rows and columns. For unordered data, spectral techniques can be used to find good orders~\cite{gionis:04:geometric}. We gave the \FindTile procedure for which we proved it finds the locally optimal tile in $\Theta(\numrows\numcols\min(\numrows,\numcols))$. 

Experiments showed \Mondrian discovers high-quality tile trees, providing succinct description of binary data. Importantly, by their hierarchical shape and small size, these models are easily interpreted and analysed by hand. 

Future work includes optimising the encoded cost by mining tiles and orders at the same time, as opposed to using ordering techniques oblivious to the target.

\section*{Acknowledgements}
Nikolaj Tatti and Jilles Vreeken are supported by Post-Doctoral Fellowships of the Research Foundation -- Flanders (\textsc{fwo}).

\bibliographystyle{plain-initials}
\bibliography{bib/abbrev,bib/bib-jilles}

\providecommand{\noopsort}[1]{}
\begin{thebibliography}{10}

\bibitem{agrawal:94:fast}
R.~Agrawal and R.~Srikant.
\newblock Fast algorithms for mining association rules.
\newblock In {\em VLDB}, pages 487--499, 1994.

\bibitem{bringmann:07:chosen}
B.~Bringmann and A.~Zimmermann.
\newblock The chosen few: On identifying valuable patterns.
\newblock In {\em ICDM}, pages 63--72, 2007.

\bibitem{calders:07:mining}
T.~Calders, N.~Dexters, and B.~Goethals.
\newblock Mining frequent itemsets in a stream.
\newblock In {\em ICDM}, pages 83--92. IEEE, 2007.

\bibitem{cover:06:elements}
T.~M. Cover and J.~A. Thomas.
\newblock {\em Elements of Information Theory}.
\newblock Wiley-Interscience New York, 2006.

\bibitem{debie:11:dami}
T.~De~Bie.
\newblock Maximum entropy models and subjective interestingness: an application
  to tiles in binary databases.
\newblock {\em Data Min.\ Knowl.\ Disc.}, 23(3):407--446, 2011.

\bibitem{fortelius:06:spectral}
M.~Fortelius, A.~Gionis, J.~Jernvall, and H.~Mannila.
\newblock Spectral ordering and biochronology of european fossil mammals.
\newblock {\em Paleobiology}, 32(2):206--214, 2006.

\bibitem{geerts:04:tiling}
F.~Geerts, B.~Goethals, and T.~Mielik\"{a}inen.
\newblock Tiling databases.
\newblock In {\em DS}, pages 278--289, 2004.

\bibitem{gionis:04:geometric}
A.~Gionis, H.~Mannila, and J.~K. Sepp{\"a}nen.
\newblock Geometric and combinatorial tiles in 0-1 data.
\newblock In {\em PKDD}, pages 173--184. Springer, 2004.

\bibitem{grunwald:07:book}
P.~Gr\"{u}nwald.
\newblock {\em The Minimum Description Length Principle}.
\newblock MIT Press, 2007.

\bibitem{han:07:freqpat}
J.~Han, H.~Cheng, D.~Xin, and X.~Yan.
\newblock Frequent pattern mining: Current status and future directions.
\newblock {\em Data Min.\ Knowl.\ Disc.}, 15, 2007.

\bibitem{hanhijarvi:09:tell}
S.~Hanhij{\"a}rvi, M.~Ojala, N.~Vuokko, K.~Puolam{\"a}ki, N.~Tatti, and
  H.~Mannila.
\newblock Tell me something {I} don't know: randomization strategies for
  iterative data mining.
\newblock In {\em KDD}, pages 379--388. ACM, 2009.

\bibitem{konto:10:sdm}
K.-N. Kontonasios and T.~{De Bie}.
\newblock An information-theoretic approach to finding noisy tiles in binary
  databases.
\newblock In {\em SDM}, pages 153--164. SIAM, 2010.

\bibitem{vitanyi:93:book}
M.~Li and P.~Vit\'{a}nyi.
\newblock {\em An Introduction to Kolmogorov Complexity and its Applications}.
\newblock Springer, 1993.

\bibitem{mampaey:11:tell}
M.~Mampaey, N.~Tatti, and J.~Vreeken.
\newblock Tell me what {I} need to know: Succinctly summarizing data with
  itemsets.
\newblock In {\em KDD}, pages 573--581. ACM, 2011.

\bibitem{miettinen:08:discrete}
P.~Miettinen, T.~{Mielik\"ainen}, A.~Gionis, G.~Das, and H.~Mannila.
\newblock The discrete basis problem.
\newblock {\em IEEE TKDE}, 20(10):1348--1362, 2008.

\bibitem{mitchell-jones:99:atlas}
A.~Mitchell-Jones, G.~Amori, W.~Bogdanowicz, B.~Krystufek, P.~H. Reijnders,
  F.~Spitzenberger, M.~Stubbe, J.~Thissen, V.~Vohralik, and J.~Zima.
\newblock {\em The Atlas of European Mammals}.
\newblock Academic Press, 1999.

\bibitem{myllykangas:06:dna}
S.~Myllykangas, J.~Himberg, T.~B\"{o}hling, B.~Nagy, J.~Hollm\'{e}n, and
  S.~Knuutila.
\newblock {DNA} copy number amplification profiling of human neoplasms.
\newblock {\em Oncogene}, 25(55):7324--7332, 2006.

\bibitem{pensa:05:bicluster}
R.~G. Pensa, C.~Robardet, and J.-F. Boulicaut.
\newblock A bi-clustering framework for categorical data.
\newblock In {\em PKDD}, pages 643--650. Springer, 2005.

\bibitem{tatti:11:order}
N.~Tatti.
\newblock Are your items in order?
\newblock In {\em SDM}, pages 414--425. SIAM, 2011.

\bibitem{tatti:08:decomposable}
N.~Tatti and H.~Heikinheimo.
\newblock Decomposable families of itemsets.
\newblock In {\em ECML PKDD}, pages 472--487. Springer, 2008.

\bibitem{vreeken:11:krimp}
J.~Vreeken, M.~{van Leeuwen}, and A.~Siebes.
\newblock \textsc{Krimp}: Mining itemsets that compress.
\newblock {\em Data Min.\ Knowl.\ Disc.}, 23(1):169--214, 2011.

\end{thebibliography}

\ifapx
\appendix
\section{Proofs}
\label{sec:apx}

\begin{proof}[of Proposition~\ref{prop:gain}]
Let $i = \tid{X; \tree{T}}$. It follows directly from the definition that for
any tile $Z \in \tree{T}$ we have $\tid{Z; \tree{T}'} = \tid{Z; \tree{T}}$ if $\tid{Z;
\tree{T}} < i$, while $\tid{Z; \tree{T}'} = \tid{Z; \tree{T}} + 1$ if $\tid{Z;
\tree{T}} \geq i$. We also have $i = \tid{Y; \tree{T}'}$. This implies that $\cells{Q; \tree{T}} = \cells{Q; \tree{T}'}$ for any $Q \in \tree{T}$
such that $Q \neq X$, which implies that $\lc{D \mid Q, \tree{T}} = \lc{D \mid Q, \tree{T}'}$.
Consequently, we have
\[
\begin{split}
	\lc{D, \tree{T}'} - \lc{D, \tree{T}} & = \lc{D \mid Y, \tree{T}'} + \lc{D \mid X, \tree{T}'} - \lc{D \mid X, \tree{T}} +  \lc{Y \mid \tree{T}'} \\
	& = \ent{u, v} + \lc{D \mid X, \tree{T}'} - \ent{o, z} + \lc{Y \mid \tree{T}'}\quad.
\end{split}
\]
Since $\cells{Y; \tree{T}'} \subseteq \cells{X; \tree{T}}$,
we have $\ones{X; \tree{T}'} = o - u$ and $\zeroes{X; \tree{T}'} = z - v$.
\qed\end{proof}

\begin{lemma}
\label{lem:inequality}
Define $g(x, y, o, z) = \ent{x, y} + \ent{o - x, z - y}$.
Assume $8$ non-negative numbers $r_p$, $r_n$, $s_p$, $s_n$, $t_p$, $t_n$, $o$, $z$. 
Assume that $t_p + t_n > 0$ and $s_p + s_n > 0$ and $t_p / (t_p + t_n) \geq s_p / (s_p + s_n)$.
Write $u = r_p + s_p$ and $v = r_n + s_n$.
Assume that $u + t_p \leq o$, $v + t_n \leq z$, and $u / (u + v) > o / (o + z)$.
Let $q = g(u, v, o, z)$. Then either $g(r_p, r_n, o, z) < q$ or $g(u + t_p, v + t_n, o, z) \leq q$.
\end{lemma}

\begin{proof}
Assume that $g(r_p, r_n, o, z) \geq q$.
Note that because of $u / (u + v) > o / (o + z)$ we must have $u > 0$ and $v < z$.
Assume that $0 < v$ and $u < o$.
Define 
\[
\begin{split}
	A = -\log \frac{u}{u + v}&,\ 
	B = -\log \frac{v}{u + v},\  \\
	C = -\log \frac{o - u}{o + z - u + v}&,\ 
	D = -\log \frac{z - v}{o + z - u + v}\quad.
\end{split}
\]
Define $h(x, y) = xA + yB + (o - x)C + (z - y)D$.
Since $h(x, y)$ is essentially the length of encoding with possibly sub-optimal codes, it follows that
\begin{equation}
\label{eq:bound}
	g(x, y, o, z) \leq h(x, y) \text{ and } g(u, v, o, z) = h(u, v)\quad.
\end{equation}

Define $f(x, n) = n(x(A - C) + (1 - x)(B - D))$.
We have the following identity between $h$ and $f$,
\[
\begin{split}
	h(x, y) - h(u, v) & = (x - u)(A - C) + (y - v)(B - D) \\
	& = f( (x - u) /(x + y - u - v) , x + y - u - v)\quad.
\end{split}
\]
Since $u / (u + v) > o / (o + z)$, we have $A < C$ and $D < B$,
which implies that $f(x, n)$ is non-increasing w.r.t to $x$ for any $n \geq 0$.

Let $m = s_p + s_n$ and define $w = s_p / m$. We have
\[
	0 \geq g(u, v, o, z) - g(r_p, r_n, o, z) \geq h(u, v) - h(r_p, r_n) = -f(w, -m) = f(w, m)\quad.
\]

This implies that $f(w, n) \leq 0$ for any $n \geq 0$. Since $t_p / (t_p + t_n) \geq w$,
we have $f(t_p / (t_p + t_n), t_p + t_n) \leq 0$. This implies that
\[
\begin{split}
	g(u + t_p, v + t_n, o, z) - g(u, v, o, z) & \leq h(u + t_p, v + t_n) - h(u, v) \\
	& = f(t_p / (t_p + t_n), t_p + t_n) \leq 0\quad.
\end{split}
\]

Assume now that $v = 0$, this will make $B = \infty$. However, we can repeat
the proof as long as $D < B$ and Eq.~\ref{eq:bound} is satisfied. This can be
done if we select $B$ high enough, say $B = \max g(x, y, o, z)$, where $0 \leq
x \leq o$ and $0 \leq y \leq z$. The argument is similar for case $u = o$.
\qed\end{proof}

\begin{proof}[of Proposition~\ref{prop:inequality}]
We will only show the case that there exist $a$ and $b$ such that $\cost{u, v} \leq \cost{i, j}$
and $a$ is a head border of $b$. The proofs for other cases are similar.

Assume that $i$ is not a head border of $j$. There exist
indices $1 \leq x < i \leq y \leq j$ such that $\freq{x, i - 1} \geq \freq{i, y}$.
Let
$r_p = \cnt{y + 1, j; p}$,
$r_n = \cnt{y + 1, j; n}$,
$s_p = \cnt{i, y ; p}$,
$s_n = \cnt{i, y ; n}$,
$t_p = \cnt{x, i - 1 ; p}$,
$t_n = \cnt{x, i - 1 ; n}$.
Then the conditions in Lemma~\ref{lem:inequality} are satisfied. Hence either
$\cost{y + 1, j} < \cost{i, j}$ or $\cost{x, j} \leq \cost{i, j}$.
Note that in the first case we must have $y + 1 \leq j$ since $\cost{i, j} \leq \cost{j + 1, j} = \ent{o, z}$.
We can now reset $a = y + 1$ or to $a = x$ if it is the second case, and repeat the argument.
Note that during each step we either decrease the score or move $i$ to the left. Since there
are only finite number of possible scores, this process will eventually stop and we have found
$a = i$ that is a head border of $b = j$.
\qed\end{proof}

\fi

\end{document}